\documentclass[reqno,12pt,letterpaper,final]{amsart}

\usepackage{fullpage}
\usepackage{bbm}
\usepackage{graphicx}
\usepackage{upref}
\usepackage{enumerate}
\usepackage{latexsym}
\usepackage{verbatim}

\usepackage{amsmath}
\usepackage{amsthm}
\usepackage{amsfonts}

\usepackage{amssymb}
\usepackage{varioref}
\usepackage[ansinew]{inputenc}

\newtheorem{theorem}{Theorem}
\newtheorem{lemma}[theorem]{Lemma}

\newtheorem{corollary}[theorem]{Corollary}

\newtheorem*{thm1*}{Theorem 1}
\newtheorem*{thm2*}{Theorem 2}
\newtheorem*{prop*}{Proposition}

\newtheorem{prop}[theorem]{Proposition}

\newtheorem{definition}{Definition}

\allowdisplaybreaks[2]

\begin{document}
\title {Computation in anonymous networks}

\author{Elchanan Mossel $^{1}$,  Anupam Prakash $^{2}$ and Gregory Valiant $^{3}$}

\thanks{$^{1}$ UC Berkeley, supported by  grants NSF-DMS-1106999  and DOD-ONR-N000141110140 \emph{E-mail:} mossel@stat.berkeley.edu}
\thanks{$^{2}$ Computer science division, UC Berkeley,  \emph{E-mail:} anupamp@eecs.berkeley.edu}
\thanks{$^{3}$ Computer science division, UC Berkeley,  \emph{E-mail:} gregory.valiant@gmail.com}

\begin{abstract}
We identify and investigate a computational model arising in molecular computing, social computing and sensor network. 
The model is made of of multiple agents who are computationally limited and posses no global information. 
The agents may represent nodes in a social network, sensors, or molecules in a molecular computer. 

Assuming that each agent is in one of $k$ states, we say that {\em the system computes} $f:[k]^{n} \to [k]$ if all agents eventually converge to the correct value of $f$. We present number of general results 
characterizing the computational power of the mode. We further present 
protocols for computing the plurality function with $O(\log k)$ memory and for approximately counting the number of nodes of a given color with $O(\log \log n)$ memory, where $n$ is the number of agents in the networks. These results are tight. 
\end{abstract}

\maketitle
\section{Introduction}
In this paper we present a common framework for studying computational problems involving multiple limited memory agents.
The framework is motivated by models of molecular computing, static and dynamic social networks and sensor networks. 
We study the computational power of the new framework: which functions can be computed by a molecular computer? by a dynamic social network? or by a sensor network? 

Given a function $f$ we interested in designing the basic computational model of identical bounded memory agents so that the resulting system will compute the desired function for {\em any} realization of the interaction between the agents. For example, in the case of molecular computers, we want the output to be correct answer no matter what order do the molecules interact with one another. Similarly, in the case of dynamic social networks, we wish that the system will compute the correct output no matter what the dynamics of the network and interaction is 
(as long as minimal connectivity requirements are preserved). 

The computational task is therefore to evaluate $f:[k]^{n} \to [k]$
where the inputs to $f$ are drawn from $[k]$ and are distributed among the $n$ agents in the network. 
We investigate the computational power of networks of anonymous agents where agents have no information about the size of the network, its structure 
or the identities of other agents.  In such a setting, it is natural to restrict the family of functions, $f,$ considered to be \emph{symmetric} functions, that only depend on the number of agents with each type.

The interaction between agents is modeled by a connectivity graph $G$ (that may change with time). The model is further parameterized 
by the amount of memory $m$ available to each agent and the rate of communication $\nu$ between neighbors.  
The interaction between neighbors is modeled as a Poisson process
of rate $\nu$.

\subsection{Motivation and related research}  
\subsubsection{Computation on Social Networks} 
Groundbreaking work by  Latan\'{e} and
L'Herrou~\cite{Latane:96} introduced a computational model where agents belong to a social network and their goal is to compute some specific function observing only the state of their neighbors. A key assumption in their experiment is the very limited signaling power (state space) of individual agents. A key insight of later work~\cite{KSM:06,KJTW:09,ECPW:09} is that in order to understand the behavior of human test subject, it is first necessary to understand the computational power of these models as abstract computational model. 

In~\cite{MosselSchoenebeck:10} a formal model for studying computation on social networks is introduced and the the memory complexity of computing  consensus and majority on connected networks is determined. Later work studies the convergence speed in terms of properties of the graph~\cite{DraiefVojnovic:12} (perhaps also~\cite{BabaeeDraief:12}).

\subsubsection{Sensor networks} 
Sensor networks share the main feature of the computational model above: the sensor do not have global information. Further, they are computationally limited. The question of the computational power of sensor networks is thus a natural one. While traditionally, the research have focused on computing real valued functions (such as averages) a recent research direction involves computing discrete valued functions such as the majority function~\cite{BeThVe:09}. 

\subsubsection{Molecular programing} 
Essentially the same model is studied in the area of molecular computing. Here, using chemical reaction networks with DNA, the goal is to compute a function of initial state (number of molecules of each type) using pairwise interaction between molecules~~\cite{QiSoWi:10,SoSeWi:10}.
A key goal in this area is the minimize the number of molecules involved in the interaction. This is the same as the amount of memory.

\subsubsection{Distributed Computing}
The model under discussion may sound similar to standard models in distributed computing, see e.g.~\cite{Peleg:00}. 
However, the model are quite different. 
The first difference is that in distributed computing it is never assumed that the nodes have very limited memory. In particular, our main interest is when the memory per agents is $o(\log n)$, where $n$ is the size of the network. Furthermore we cannot assume that nodes have unique i.d's. This prohibits the use of many of the standard algorithms in distributed computing. A second difference is that    distributed computing models have broadcast time on edges which does not follow a distribution. Instead it is only guaranteed to be bounded. Our models will assume that the broadcast times follow a Poisson process. 

\section{Model}

\subsection{Terminology} 
For the reminder of the paper we will assume that we are given a {\em fixed connect social network} with $n$ nodes. 
Our goal is to compute a function $f : [k]^n \to [k]$ where we view the input as the $n$ values stored at the $n$ different nodes. 
All of our results extend to dynamic networks as long as connectivity is maintained. To simplify notation we will assume that the network is fixed. The model from molecular computation correspond to the (easier) case of the complete graph.

\subsection{Results:}  

We are interested in computing symmetric functions $f : [k]^n \to [k]$. In the case where $k=2$ if $f:\{0,1\}^{n} \to \{0,1\}$ is symmetric then it just depends on the number of $0$'s in the network. 
 For the $k$-ary setting where $k$ different colors/properties are present, computing the plurality color and other comparison statistics are natural extensions of majority and symmetric functions for the binary case. Computation of symmetric functions is thus motivated by the problem of estimating
statistics over networks.

Given $O(\log n)$ memory of per agent---with high probably nodes can generate unique labels for themselves; computation over such networks has been studied extensively int the distributed computing literature.  Here, we seek to understand the the class of functions that can be computed by ``anonymous'', networks, and thus we restrict ourselves to the regime in which the available memory is $o(\log n).$   In the setting in which nodes are originally assigned binary labels, and the goal is to compute some symmetric function of the number of $1$'s, our protocols  use $O(\log \log n)$ memory per agent.  In the $k$-ary setting, we describe protocols that use $O(\log k)$ memory per agent. 

A counting argument shows that there exist symmetric functions that require $O(\log n)$ bits of memory per node to compute, hence not all symmetric functions can be computed 
by anonymous agents. Nevertheless, we show that many of the symmetric functions that are relevant for estimating statistics of the network are computable with small amounts of memory.

Given $c$ bits of memory per
agent it is possible to compute: (i) The $c$ least significant bits in the  binary representation of $r$, where $r$ is the number of nodes with color $0$. (ii) The threshold function
$f: [n] \to \{0,1\}$ which is $1$ if $\frac{r}{n-r} \geq \frac{a}{b}$ and $0$ otherwise for integers $a,b \leq 2^{c}$.
The higher and lower order bits of $r$ can be computed with a constant amount of memory, the computation of middle bits 
might require more memory. Interestingly, we show that every bit in the binary representation of $r$ can be computed with $O(\log \log n)$ memory.
\begin{theorem}
Every bit in the binary representation of $r$ can be computed by a network with $\log \log n+2$ bits of memory per agent.
\end{theorem}
\noindent The theorem can be applied to the $k$-ary setting for constant $k$ to approximately count the number of nodes of the different colors
and estimate the parameters like variance and entropy for the probability distribution over colors  with $O(\log \log n)$ memory. 

The majority protocol of~\cite{MosselSchoenebeck:10} can be used to implement $Max/Min$ gates for comparing the number of nodes of two different colors. 
Implementing each gate requires a constant amount of memory, however composing the gates to evaluate circuits is non trivial due to the 
asynchronous nature of the model. We prove a composability lemma showing that the outputs of comparison gates can be composed 
in the network setting with constant amount of memory per gate.
\begin{theorem} 
A circuit of depth $d$ consisting of comparison gates can be evaluated by a network with $O(d)$ memory per agent. 
\end{theorem} 
\noindent Consider the problem of computing the plurality color for a network with $k$ different colors, it is easy to compute plurality with
 $O(k^{2})$ or $O(k)$ memory per node using the majority protocol of~\cite{MosselSchoenebeck:10}, while the lower bound on memory is $O(\log k)$ 
 as $\log k$ bits are needed to store a color from $k$. Simulating the comparison tree of depth $\log k$ for computing the 
 maximum of $k$ numbers over a network using theorem 2 we obtain: 
 \begin{corollary}
The plurality function over $k$ colors can be computed by a network with $O(\log k)$ bits of memory per agent.
\end{corollary}
\noindent All the protocols that we describe for anonymous agents 
are independent of the geometry of the graph and converge in time $O(n^{3})$. The protocols are robust against 
change in network structure and remain useful in the setting where agents are online for limited time periods. 


\section{Preliminaries} 
\subsection{Bounded memory dynamics:} The bounded memory dynamics model is formalized as a three tuple $(G, \Sigma, T)$, 
\begin{enumerate} 
\item $G(V,E)$ is the underlying graph, for an edge $(i,j) \in E(G)$, interaction between $(i,j)$ is modeled as a Poisson process of rate $\nu$. 
\item $\Sigma = \Sigma_{e} \times \Sigma_{i}$ is the memory state of the agents, $\Sigma_{e}$ is the external memory that is visible to other 
agents in the network while $\Sigma_{i}$ is the private memory.  
\item $T$ is a possibly randomized of transition rule, $T: \Sigma \times \Sigma_{e} \times r \to \Sigma$. The transition rule takes as input the current state, 
the state of the neighbor and random bits represented by $r$ and produces the next state. 
\end{enumerate} 
The difference between $\Sigma_{i}$ and $\Sigma_{e}$ can be used to impose restrictions on the amount of communication allowed during an interaction. 
 We bound the total amount of memory $m=|\Sigma_{i} + \Sigma_{e}|$ so the distinction between external and internal memory is not important for us. 
A single step of computation in the bounded memory dynamics model is the following: the edge $(i,j)$ is activated, and agents $i$ and $j$ interact and update 
their states according to the transition rule $T$. 

\subsection{Computation:} Computation is modeled as the task of evaluating a function $f: [k]^{n} \to [k]$ where the inputs to $f$ are distributed among the 
$n$ agents in the network. The network computes the function $f$ if for all possible inputs all the agents in the network eventually converge to the correct value 
of $f$. A protocol specified by the transition rule $T$ is said to compute $f$ if all networks using the rule $T$ in the bounded memory dynamics compute $f$. 
A protocol is efficient if the expected convergence time is $poly(n)$.  

The following lemma is used to bound the expected running time 
for all the protocols presented here, the proof can be found in section 6.4.3 of the book by Aldous-Fill~\cite{AldousFill:u} on Markov chains, 
\begin{lemma} \label{cat} 
The expected time in which two independent random walks on a graph $G$ meet is $O(n^{2})$. 
\end{lemma}

\section{Binary symmetric functions} 

Agents are assigned values from $\{0,1\}$ in the binary setting. A node with value $0$ is called a red node and $r$ denotes the total number of red nodes in 
the network. A function $f:\{0, 1\}^{n} \to \{0,1\} $ is symmetric if it depends only on the number of red nodes in the network. The total number of 
symmetric functions $f:[n] \to \{0,1\}$ is $2^{n}$. 
\subsection{Constant memory protocols:} 
The only deterministic transition rule with $1$ bit of memory computes the $OR$ function; that is, the output is $1$ if the network has at least one red/blue node. 
Randomized transition rules with $1$ bit of memory per node can be used to converge to consensus and find a bipartition if the 
network is bipartite. 

With two bits of memory it is possible to compute the parity and majority functions. Protocols for reaching consensus and finding the majority with 
$2$ bits of memory per node were presented in~\cite{MosselSchoenebeck:10}. We next present the parity protocol and the threshold protocol, which require a constant amount of memory. 
The protocols are simple but useful as building blocks of more sophisticated protocols. 

\begin{prop} 
The $c$ least significant bits in the binary representation of $r$ can be computed with $c+1$ bits of memory. 
\end{prop} 
\begin{proof} 
The memory for each node is partitioned as $\Sigma=(C,b)$ where counter $C$ stores the $c$ least significant bits in the binary representation of $r$ and bit $b$ indicates active/passive status. 
The initial states are $(1,1)$ for red nodes and $(0,0)$ otherwise. The interaction rules are the following: (i) When two active nodes meet one passes information to the other and becomes passive $(c_{1}, 1), (c_{2}, 1) \to (c_{1}+c_{2}, 1) (c_{1}+c_{2}, 0)$.  (ii) A passive node copies the counter of an active node and active/passive status is swapped $(c_{1}, 0) (c_{2}, 1) \to (c_{2}, 1), (c_{2}, 0)$. 

The number of active nodes in the network decreases by $1$ when two active nodes meet and any two active nodes are guaranteed to meet in time $O(n^{2})$ by lemma \ref{cat} . 
The last active node $x$ in the network has counter $c= r \mod 2^{c}$ and all nodes converge to the correct answer after meeting $x$. The expected time for convergence is $O(n^{3})$. 
\end{proof} 
\noindent With $\log n+1$ bits of memory per node, the entire binary representation of $r$ can be computed and stored by all the nodes using the parity protocol. 
A counting argument shows that at least $O( \log n)$ bits of memory per node are necessary to compute every symmetric function. 
 
\begin{prop}
  There exist functions $f:[n] \rightarrow \{0,1\}$ that cannot be computed by a network of identical nodes where each node has $k \le \frac{\log n}{3}$ bits of memory. (Assuming the amount of randomness is some small constant.)
\end{prop}
\begin{proof}
  The total number of symmetric functions $f:[n] \rightarrow \{0,1\}$ is $2^n.$  The number of different protocols with $k$ bits of memory is at most $(2^{2k})^{2^{2k+c}},$ where $c$ is the number of random bits they are given access to, since for each of the $2^{2k+c}$ possible combinations of their combined memory and the random bits a pair of nodes can rewrite their memories in any one of $2^{2k}$ ways.
  
  Given a network with $n$ identical nodes, since every protocol can compute at most one function, if $2k+c+\log(2k) < \log n,$ then raising everything to the power of 2, twice, yields $(2^{2k})^{2^{2k+c}} < 2^n,$ and thus there must be some symmetric function that cannot be computed with $k$ bits of memory per node.
\end{proof}

The problem of finding an explicit symmetric function that requires $\Omega(\log n)$ memory for computation remains open. 
The second basic protocol for computation on networks computes rational thresholds generalizing the majority protocol 
from~\cite{MosselSchoenebeck:10}. 

\begin{prop} 
The threshold function $f: [n] \to \{0,1\}$ which is $1$ if $\frac{r}{n-r} > \frac{a}{b}$ and $0$ otherwise for integers $a,b \leq 2^{c}$ can be computed with 
$c+2$ bits of memory. 
\end{prop} 
\begin{proof} 
The memory for each node is partitioned into $(C,b)$ where $C$ is a counter with $c+1$ bits and bit $b$ indicates weak/strong status. The counter $c$ is initialized to $+b$ for red nodes 
and $-a$ otherwise, the sign of the counter being the color, all nodes start as strong. The interaction rules are the following: (i) If strong nodes of opposite color meet cancellation occurs, 
$(c_{1}, 1) (c_{2}, 1) \to (c_{1}+c_{2}, 1) (0, 0)$ if $|c_{1}|>|c_{2}|$, if $|c_{1}|=|c_{2}|$ the nodes change to $(0,0), (0,1)$. (ii) A weak node
copies the counter of a strong node and weak/strong status is swapped $(c_{1}, 0) (c_{2}, 1) \to (c_{1}, 1), (c_{1}, 0)$. 

Any two nodes in the network with opposite colors eventually meet, so finally the network has strong nodes of at most one color. The total value of the counters is 
$br - a(n-r)$ remains invariant throughout the process, the network converges to red if $\frac{r}{n-r} > \frac{a}{b}$. If equality holds the network converges to any one 
of the colors. The expected time for convergence is $O(n^{3})$ as $n$ collisions suffice for convergence. 
\end{proof} 
The parity and threshold protocols show that with $c+2$ bits of memory the network can compute symmetric functions 
that depend on the $c$ most or least significant bits in the binary representation of $r$. Extracting information about the
middle bits in the binary representation of $r$ using these protocols requires $O(\log n)$ memory.

\subsection{Approximate counting:} 
Any bit in the binary representation of $r$ can be computed with $O(\log \log n)$ bits of memory, the protocol can be used to estimate $r$ 
within a constant factor.  
\begin{thm1*} 
Every bit in the binary representation of $r$ can be computed by a network with $\log \log n+2$ bits of memory per agent. 
\end{thm1*} 
\begin{proof} 
We describe the protocol for computing the $j$-th bit of the binary representation of $r$. The memory state of every node is partitioned into 
a three tuple $(c,b, l)$ where the first bit $c$ stores the color, the second bit $b$ stores the active/passive status and the remaining  $\log \log n$ bits store the level $l$, where $l$ is the 
bit in the binary representation of $r$ that the node is currently keeping track of. The protocol starts with all red nodes active $(b=1)$ and at level $0$. 

Nodes with $b=1$ are active and participate in the computation while nodes with $b=0$ are passive and copy the color of an active node at level $j$ on meeting one. When a passive node 
meets an active node, the identities of the two nodes are swapped to ensure that the active nodes perform a random walk. Active nodes at the  
same level interact on meeting, and there is no interaction between active nodes at different levels. 
The interaction rules for the meeting of two active nodes $(c_{1}, 1, l)$ and $(c_{2}, 1, l)$ are given as follows: 
\begin{eqnarray*}
(c_{1}, 1, l), (c_{2}, 1, l) &\to& (0, 1, l), (1, 1, l+1)   \hspace{70pt} \text{if $c_{1}=c_{2}=1$} \\
&\to& (c_{1}+c_{2}, 1, l) ,(0, 0, l) \hspace{60pt} \text{otherwise}
\end{eqnarray*} 

The correctness of the protocol follows as the following invariants are maintained throughout the process: 
\begin{enumerate} 
\item The number of active nodes at level $l$ decreases by $1$ for every interaction between a pair of active level $l$ nodes. The connectivity of the graph ensures 
that any two active nodes at level $l$ eventually meet, hence if there is a non zero number of level $l$ nodes at any stage of the 
protocol eventually there will be exactly one node at level $l$. 
\item The level of an active node is a monotonically increasing function, let $l_{i}$ be the number of nodes that reach level $i$ at some stage. 
\item If $l_{i}=2k+r_{i}$, then exactly $k$ nodes reach level $i+1$ and the last active node at level $i$ has color $r_{i}$. 

Applying to level $0$, we find that the last active node at level $0$ stores $r \mod 2$ and the number of active nodes that reach level $1$ is $\lfloor r/2 \rfloor$. Applying the claim iteratively 
for higher levels it follows that the last active node at level $i$ stores the $i$-th bit in the binary representation of $r$. Passive nodes answer $0$ by default covering the case where no 
node reaches level $j$ and copy the value of an active node at level $j$ otherwise. Eventually every node in the network will output the correct value of $r_{j}$. 

\end{enumerate} 
The expected convergence time is $O(n^{3})$ by lemma \ref{cat} as the protocol converges after $O(r+n)$ collisions. 
\end{proof} 
The above protocol can be used to estimate the number of red nodes in the network within a factor of $2$, by  
having each node output the highest value of $j$ such that $r_{j}=1$. All nodes eventually converge to an estimate that is within a factor 
of $2$ of the number of red nodes. More generally with $\log \log n+ c+2$ bits of memory it is possible to estimate the number of red nodes 
within a factor of $1/2^{c}$. 

\section{Computing Comparison Statistics }
The binary setting generalizes to computing $k$-ary symmetric functions $f: [k]^{n} \to [k]$. 
Computing comparison statistics such as the most popular color (plurality) or the least common color are natural questions in the $k$-ary setting. 
We prove a composability lemma showing that depth $d$ circuits consisting of comparison gates can be simulated with $O(d)$ memory. 

Circuits are built out of distributed comparison gates, the distributed $\max$ gate is defined as follows, the $\min$ gate can be defined analogously, 
\begin{definition} 
$\max(a_{1}, a_{2})$: The gate acts on $a_{1}+a_{2}$ nodes where $a_{1}$ nodes have type $1$ and $a_{2}$ nodes have type $2$. Eventually $\max(a_{1}, a_{2})$ nodes 
output $1$ and the remaining $\min(a_{1}, a_{2})$ nodes output $0$. 
\end{definition} 
A protocol similar to the majority protocol from~\cite{MosselSchoenebeck:10} can be used to implement a max gate with three bits of memory per node: (i) Nodes of types $1,2$ have an initial 
configuration $(1,1)$ and $(-1,1)$ respectively. (ii) The update rules are: $(1, 1), (-1, 1) \to (0, 1) (0, 0)$ otherwise $(a, b) (x,y) \to (x,y) (a, b)$. The third bit is treated as the  
output bit, there are $\min(a_{1}, a_{2})$ collisions between nodes with charge $\pm 1$, ensuring that $\max(a_{1}, a_{2})$ nodes eventually output $1$. A min gate 
is implemented similarly with the output bits initialized to $0$ and with one bit updated to $1$ every time a collision occurs.

It is easy to construct a plurality protocol requiring $O(k)$ memory by having a node of color $c$ store results of pairwise comparisons with other colors, where comparisons are implemented as in the majority protocol. The minimum memory required for a $k$-ary protocol is $O(\log k)$ as $\log k$ bits are needed to store a color from $[k]$. We show that a circuit built out of $\min$ and $\max$ gates can be evaluated with memory proportional to the circuit depth and thus obtain a protocol for computing the plurality function  
with $6\log k$ bits of memory as a special case.

The composability lemma shows that computations involving $\max$ and $\min$ gates can be composed with one extra bit of memory for each level. We 
prove the lemma for $\max$ gates, the proof for $\min$ gates is analogous,  
\begin{lemma} \label{comp} 
The gate $G=\max(A,B)$ can be evaluated with one additional bit of memory per node if inputs to the gate are outputs of other $\max$ gates. 
\end{lemma}  
\begin{proof} 
The initial configuration has $A$ nodes with charge $+1$ and $B$ nodes with charge $-1$, the action of the $\max$ gates in lower levels of the circuits effectively leads to a  
decrease in $A$ and $B$, the final result of the computation should reflect these changes. 

One extra bit $m$ called the mark bit is stored for all the $A+B$ nodes involved in the computation, the mark bits are initialized to $0$. A collision between nodes $(x,y)$ at a lower level 
occurs only if $m_x=0$, and $m_x$ is set to $1$ if the collision occurs. 

Wlog assume that node $x$ had initial charge $+1$, a set mark bit $m_{x}=1$ indicates that the total charge should be decreased by $1$ unit. The bit $m_{x}$ is reset to $0$ 
in one of the following ways depending on $C(x)$, the charge on $x$ at the input to $G$: (i) If $C(x)=1$ then it is set to $0$ both at the input and output of $G$. (ii) If $C(x)=0$,  $C(x)$ is set to $-1$ with no change at the output of $G$. (iii) If $C(x)=-1$ wait until the situation in (i) or (ii) occurs, the connectivity of the graph ensures that the mark bit eventually 
gets reset. 

Let $a$ and $b$ denote the output values of the max gates that are inputs to the gate $G$, and $(c_{1}, c_{2})$ and $(d_{1}, d_{2})$ be the number of updates of type $(i)$ and $(ii)$ above to nodes of type $A$ and $B$. When the inputs gates to $G$ have converged to their true values $(a,b)$ we have, 
\begin{eqnarray*}
A = c_{1}+ c_{2} + a \\ B= d_{1}+ d_{2} + b 
\end{eqnarray*}
The number of subsequent collisions between the $A+B$ nodes that are inputs to $G$ is equal to, 
\begin{eqnarray}\label{col} 
\min(A - c_{1} + d_{2}, B- d_{1} + c_{2}) &=& \min ( c_{2}+d_{2}+a, c_{2}+d_{2}+b)\notag  \\
&=& c_{2}+ d_{2} + \min(a,b) 
\end{eqnarray}
The number of $1$s at the output is obtained by subtracting the number of collisions at the lower level from $A+B- c_{1} - d_{1}= c_{2} + d_{2} + a+ b$. 
By equation \eqref{col} the number of $1$s at the final output is $a+b - \min(a,b) = \max(a,b)$, showing that the gate $G$ is evaluated correctly.

\end{proof} 

The composability lemma applied to a circuit of depth $d$ shows that depth $d$ comparison circuits can be implemented with $O(d)$ bits of memory, thus 
proving theorem $2$. The depth $\log k$ circuit for computing the maximum of $k$ numbers yields a protocol for computing plurality with 
$O(\log k)$ memory, 
\begin{theorem} 
The plurality function over $k$ colors can be computed by a network with $6\log k$ bits of memory per node. 
\end{theorem}
\begin{proof} 
The complete binary tree $T$ with $k$ leaves has depth $\log k$ and computes the maximum of numbers $a_{1}, a_{2}, \cdots , a_{k}$ if 
$a_{i}$ is the input at the $i$-th leaf and the interior nodes of the tree are max gates. A node of color $c$ stores $4$ bits for every gate 
on the path from the root to the leaf of color $c$, and the composability lemma is used to simulate the comparison tree. 

The nodes store the initial and final color, the total memory being $6 \log k$ bits. A node of the plurality color 
eventually wins all the comparisons performed on the path from the root to the plurality color. If node $x$ interacts with $y$ such that $y$ has won all the 
rounds of comparison then $x$ copies the color of $y$ in the final color register. Eventually the plurality color is the only color that can be copied, so all nodes 
converge to the plurality color. 

The number of collisions required for the protocol to converge is $O(n)$ so the expected convergence time is $O(n^{3})$ by lemma \ref{cat}.

\end{proof} 

\section{Concluding remarks} 
The results demonstrate that the tasks of voting to find out the most popular choice out of $k$ and counting the number of nodes with a given property 
can be carried out by networks of anonymous agents with no information about the network structure.  It is surprising that a weak computational model 
suffices to compute aggregate statistics in a distributed manner, the protocols are efficient but the $O(n^{3})$ convergence time limits their utility as 
voting schemes. 

We restricted ourselves to networks with $o(\log n)$ memory as unique node identities can be established with high probability with $O(\log n)$ 
memory per agent. Computation over networks with unique identifiers has been studied extensively in the distributed computing literature and 
some properties of the underlying network can be learnt in this setting. We conjecture that no non trivial graph properties can be learnt in the anonymous setting. 

\subsection{Acknowledgement}
The first author thanks Goerg Seeling for references in the area of molecular computing and D. Moez for references to his work. 

\bibliographystyle{plain}
\bibliography{all.bib,my.bib}

\end{document}